\documentclass[conference,a4paper]{IEEEtran}
\pdfoutput=1

\usepackage{pgfplots}
\pgfplotsset{compat=newest}
\usepackage{tikz}
\usetikzlibrary{arrows,matrix,positioning,patterns}
\usepackage[utf8]{inputenc}
\usepackage[innermargin=0.6in,outermargin=0.60in,top=0.35in,bottom=0.41in]{geometry}
\usepackage[english]{babel}
\usepackage[T1]{fontenc}
\usepackage{epsfig}
\usepackage{amsmath, amssymb, amsbsy}
\usepackage{mathdots}
\usepackage{xspace}
\usepackage[noend]{algpseudocode}
\usepackage[linesnumbered,ruled,vlined,titlenumbered]{algorithm2e}
\usepackage{algorithmicx}
\usepackage{color}
\usepackage{cite}
\usepackage{booktabs}
\usepackage{verbatim}
\usepackage{url}
\usepackage{lipsum}
\usepackage{enumitem}
\usepackage{colortbl}
\usepackage{amsthm}
\usepackage{dblfloatfix}
\usepackage{verbatim}
\makeatletter
\newcommand\footnoteref[1]{\protected@xdef\@thefnmark{\ref{#1}}\@footnotemark}
\makeatother

\usepackage{multicol}

\newtheorem{theorem}{Theorem}

\newtheorem{corollary}[theorem]{Corollary}

\usepackage[final,tracking=true,kerning=true,spacing=true,factor=1100,stretch=10,shrink=20]{microtype}

\newenvironment{mymatrix}{\begin{bmatrix}} {\end{bmatrix} }

\def\ve#1{{\mathchoice{\mbox{\boldmath$\displaystyle #1$}}%
              {\mbox{\boldmath$\textstyle #1$}}%
              {\mbox{\boldmath$\scriptstyle #1$}}%
              {\mbox{\boldmath$\scriptscriptstyle #1$}}}}

\renewcommand{\vec}[1]{\ensuremath{\boldsymbol{#1}}}

\newcommand{\Fq}{\ensuremath{\mathbb{F}_q}}

\newcommand{\Code}{\mathcal{C}}

\newcommand{\mycode}[1]{\ensuremath{\mathcal{#1}}}

\renewcommand{\H}{\ve{H}}

\renewcommand{\a}{\ve{a}}
\renewcommand{\b}{\ve{b}}
\newcommand{\e}{\ve{e}}
\renewcommand{\v}{\ve{v}}

\newcommand{\m}{\ve{m}}
\renewcommand{\c}{\ve{c}}
\renewcommand{\e}{\ve{e}}
\renewcommand{\m}{\ve{m}}
\newcommand{\G}{\ve{G}}
\renewcommand{\P}{\ve{P}}
\newcommand{\R}{\ve{R}}

\newcommand{\F}{\mathbb{F}}
\usepackage{graphics} 
\usepackage{epsfig} 

\usepackage{times} 
\usepackage{amsmath} 
\usepackage{amssymb}  
\usepackage{cite}
\usepackage{soul}
\usepackage{pgfplots}
\usepackage{pgfplotstable}
\usepackage[mathscr]{eucal}
\usepackage{listings}

\newtheorem{cor}{Corollary}

\newtheorem{const}[cor]{Construction}
\newtheorem{thmmystyle}{Theorem}

\usepackage{tikz}
\usetikzlibrary{shapes}
\usetikzlibrary{shapes.multipart,chains}
\makeatletter
\newcommand{\removelatexerror}{\let\@latex@error\@gobble}

\makeatother
\makeatletter
\newcommand*{\rom}[1]{\expandafter\@slowromancap\romannumeral #1@}
\makeatother

\IEEEoverridecommandlockouts

\pgfkeys{
	/tr/rowfilter/.style 2 args={
		/pgfplots/x filter/.append code={
			\edef\arga{\thisrow{#1}}
			\edef\argb{#2}
			\ifx\arga\argb
			\else
			
			\fi
		}
	}
}

\usepackage{filecontents}

\begin{document}

\title{Bounds and Code Constructions for \\ Partially Defect Memory Cells }
\author{\IEEEauthorblockN{Haider Al Kim$^{1,2}$\thanks{This work has received funding from the German Research Foundation (Deutsche Forschungsgemeinschaft, DFG) under Grant No. WA3907/1-1. H. Al Kim has received funding from the German Academic Exchange Service (Deutscher Akademischer Austauschdienst, DAAD) under the support program ID 57381412. S.~Puchinger received funding from the European Union's Horizon 2020 research and innovation program under the Marie Sklodowska-Curie grant agreement no.~713683.}, Sven Puchinger$^{3}$, Antonia Wachter-Zeh$^{1}$}

\IEEEauthorblockA{
  $^1$Institute for Communications Engineering, Technical University of Munich (TUM), Germany\\ $^2$Electronic and Communications Engineering, University of Kufa (UoK), Iraq\\
  $^3$Department of Applied Mathematics and Computer Science, Technical University of Denmark (DTU), Denmark\\
  Email: haider.alkim@tum.de, svepu@dtu.dk, antonia.wachter-zeh@tum.de}
  }

\maketitle

\begin{abstract}
	This paper considers coding for so-called \emph{partially stuck} memory cells.
	Such memory cells can only store partial information as some of their levels cannot be used due to, e.g., wear out.
	First, we present a new code construction for masking such partially stuck cells while additionally correcting errors.
	This construction (for cells with $q >2$ levels) is achieved by generalizing an existing masking-only construction in \cite{wachterzeh2016codes} (based on binary codes)
	to correct errors as well. Compared to previous constructions in \cite{haideralkim2019psmc}, our new construction achieves larger rates for many sets of parameters.
	Second, we derive a sphere-packing (any number of $u$ partially stuck cells) and a Gilbert-Varshamov bound ($u<q$ partially stuck cells) for codes that can mask a certain number of partially stuck cells and correct errors additionally.
	A numerical comparison between the new bounds and our previous construction of PSMCs for the case $u<q$ in \cite{haideralkim2019psmc} shows that our construction lies above the Gilbert--Varshamov-like bound for several code parameters.

\end{abstract}

\begin{IEEEkeywords}
flash memories, phase change memories, defect memory, (partially) stuck cells, defective cells error correction, sphere packing bound, Gilbert-Varshamov bound 
\end{IEEEkeywords}

\section{Introduction}
The demand for reliable memory solutions and in particular for non-volatile memories such as \emph{phase-change memories} (PCMs) for different applications is steadily increasing. These memories provide permanent storage, rapidly extendable capacity, and multi-levels devices. 
However, due to increasing the number of levels while decreasing the size of the memory, it is essential to suggest sophisticated coding and signal processing solutions to overcome reliability issues. The key characteristic of PCM cells is that they can switch between two main states: an amorphous state and a crystalline state. 
PCM cells may become \emph{defect} (also called \emph{stuck}) if they fail in switching their states. This occasionally happens due to the cooling and heating processes of the cells, and therefore cells can only hold a single phase \cite{gleixner2009reliability,kim2005reliability,lee2009study,pirovano2004reliability}. The crystalline state consists of multiple substates which motivates \emph{partially stuck} type of defects to happen.

In flash memories, an electric charge might be trapped in the cell at a certain state, and the cell's status cannot be switched to be writable again. The suggested mechanism to deal with these defect memory cells whose charges are trapped is called \emph{masking}. Masking determines a codeword that matches the stuck level of the stuck memory. Therefore, it can be placed properly on the defective memory. 
In multi-level PCM cells, failure may occur at a position in both extreme states or in the partially programmable states of crystalline.
In~\cite{wachterzeh2016codes}, a cell that can only hold levels at least a certain reference level $s>0$ are called \emph{partially stuck}. For multi-level PCMs, the case $s=1$ is particularly important since this means that a cell can reach all crystalline sub-states, but cannot reach the amorphous state anymore.
In flash memories, information is stored by different charge levels. Similar to PCMs, (partial) defects can occur in flash memory cells. In order to write information in a new write, either all current levels are only increased or a whole block has to be erased. Erasing the whole block reduces the lifespan of flash memory devices.

Figure~1 shows the general idea of reliable and (partially) defect memory cells.
\begin{figure} [h]
		\scalebox{0.6}{	
		\begin{tikzpicture}
		\draw
		
		(3.25,1.8) node[anchor=north] {\small {No value can be stored}} 
		
		;
		
		\draw (6,1.4) rectangle (6.5,1.9);
		\draw(3.5,2.40) node[anchor=north] {\small {The value that cell can store}} 
		;

		\fill[blue!30!white] (6,2) rectangle (6.5,2.5);
		\draw (6,2) rectangle (6.5,2.5);
		\draw (6.3,3) node[anchor=north] {/}
		
		(3.5,3) node[anchor=north] {\small {Reliable cell stores any value}}
		;
		\draw[red,-] (6,1.4) -- (6.5,1.9);
		\draw[red,-] (6,1.9) -- (6.5,1.4);

		\end{tikzpicture}
	}
\end{figure}
\vspace{- 0.7 cm}
\begin{figure} [h] 

	\scalebox{0.70}{	
		
		\begin{tikzpicture}[baseline={(0,-0.5)}]
		
		\draw
		(-1.2,2) node[anchor=north] {Level-3 ($1+\alpha$)}
		(-1.2,1.5) node[anchor=north] {Level-2 ($\alpha$)}
		(-1.2,1) node[anchor=north] {Level-1 (1)}
		(-1.2,0.5) node[anchor=north] {Level-0 (0)}
		(0.25,2.5) node[anchor=north] {/}
		(0.75,2.5) node[anchor=north] {/}
		(1.25,2.5) node[anchor=north] {/}
		(1.75,2.5) node[anchor=north] {/}
		(2.25,2.5) node[anchor=north] {/}
		
		;
		\draw[thick,-] (0,1.5) -- (2.5,1.5);

		\draw (-0.5,3.2) node[anchor=north] {Values};
		\draw (-1.3,2.5) node[anchor=north] {Cell Levels};
		\draw[thick,-] (0,2) -- (-0.7,2.75);
		\draw[thick,-] (0,2) -- (0,2.75);
		\draw[thick,-] (0,2) -- (-0.75,2);
		\fill[blue!30!white] (0,0) rectangle (2.5,2);
		\draw (0,0) rectangle (2.5,2);
		\draw (0,0) rectangle (2.5,0.5);
		\draw (0,0) rectangle (2.5,1);
		\draw (0,0) rectangle (0.5,1.5);
		\draw (0,0) rectangle (1,1.5);
		\draw (0,0) rectangle (1.5,1.5);
		\draw (0,0) rectangle (2,1.5);
		\draw[thick,-] (2,1.5) -- (2.5,1.5);
		\draw[thick,-] (2.5,1) -- (2.5,1.5);
		\draw[thick,-] (0.5,1.5) -- (0.5,2);
		\draw[thick,-] (1,1.5) -- (1,2);
		\draw[thick,-] (1.5,1.5) -- (1.5,2);
		\draw[thick,-] (2,1.5) -- (2,2);
		
		\draw
		(3,2.5) node[anchor=north] {0}
		(3.5,2.5) node[anchor=north] {1}
		(4,2.5) node[anchor=north] {/}
		(4.39,2.5) node[anchor=north] {$\alpha$}
		(5,2.5) node[anchor=north] {$1+\alpha$};
		\draw (2.75,0) rectangle (5.25,1.5);
		\fill[blue!30!white] (2.75,0) rectangle (3.25,0.5);
		\fill[blue!30!white] (3.25,0.5) rectangle (3.75,1);
		\fill[blue!30!white] (3.75,0) rectangle (4.25,2);
		\fill[blue!30!white] (4.25,1) rectangle (4.75,1.5);
		\fill[blue!30!white] (4.75,1.5) rectangle (5.25,2);
		\draw[red,-] (2.75,0.5) -- (3.25,1);
		\draw[red,-] (2.75,1) -- (3.25,0.5);
		\draw[red,-] (2.75,1) -- (3.25,1.5);
		\draw[red,-] (2.75,1.5) -- (3.25,1);
		\draw[red,-] (2.75,1.5) -- (3.25,2);
		\draw[red,-] (2.75,2) -- (3.25,1.5);
		\draw[red,-] (3.25,0) -- (3.75,0.5);
		\draw[red,-] (3.25,0.5) -- (3.75,0);
		\draw[red,-] (3.25,1) -- (3.75,1.5);
		\draw[red,-] (3.25,1.5) -- (3.75,1);
		\draw[red,-] (3.25,1.5) -- (3.75,2);
		\draw[red,-] (3.25,2) -- (3.75,1.5);
		\draw[red,-] (4.25,0) -- (4.75,0.5);
		\draw[red,-] (4.25,0.5) -- (4.75,0);
		\draw[red,-] (4.25,1.5) -- (4.75,2);
		\draw[red,-] (4.25,2) -- (4.75,1.5);
		\draw[red,-] (4.25,0.5) -- (4.75,1);
		\draw[red,-] (4.25,1) -- (4.75,0.5);
		\draw[red,-] (4.75,0) -- (5.25,0.5);
		\draw[red,-] (4.75,0.5) -- (5.25,0);
		\draw[red,-] (4.75,0.5) -- (5.25,1);
		\draw[red,-] (4.75,1) -- (5.25,0.5);
		\draw[red,-] (4.75,1) -- (5.25,1.5);
		\draw[red,-] (4.75,1.5) -- (5.25,1);
		\draw (2.75,1.5) rectangle (3.25,2);
		\draw (2.75,1.5) rectangle (3.75,2);
		\draw (2.75,1.5) rectangle (4.25,2);
		\draw (2.75,1.5) rectangle (4.75,2);
		\draw (2.75,1.5) rectangle (5.25,2);
		\draw (2.75,0) rectangle (5.25,0.5);
		\draw (2.75,0) rectangle (5.25,1);
		\draw (2.75,0) rectangle (3.25,1.5);
		\draw (2.75,0) rectangle (3.75,1.5);
		\draw (2.75,0) rectangle (4.25,1.5);
		\draw (2.75,0) rectangle (4.75,1.5);
		;
		\draw
		
		(6,2.5) node[anchor=north] {0}
		(6.5,2.5) node[anchor=north] {1}
		(7,2.5) node[anchor=north] {/}
		(7.39,2.5) node[anchor=north] {$\alpha$}
		(8,2.5) node[anchor=north] {$1+\alpha$};
		\draw[thick,-] (6,1.5) -- (6,1.5);
		\draw (5.75,0) rectangle (7.75,1.5);
		\fill[blue!30!white] (5.75,0) rectangle (6.25,0.5);
		\fill[blue!30!white] (6.25,0.5) rectangle (6.75,1);
		\fill[blue!30!white] (6.75,0) rectangle (7.25,1.5);
		\fill[blue!30!white] (5.75,1) rectangle (7.75,1.5);
		\fill[blue!30!white] (5.75,0.5) rectangle (6.25,1);
		\draw[thick,-] (7.75,1) -- (7.75,1.5);
		\fill[blue!30!white] (7.75,0.5) rectangle (7.75,1);
		\fill[blue!30!white] (5.75,1.5) rectangle (8.25,2);
		\draw[red,-] (6.25,0) -- (6.75,0.5);
		\draw[red,-] (6.25,0.5) -- (6.75,0);
		\draw[red,-] (7.25,0) -- (7.75,0.5);
		\draw[red,-] (7.25,0.5) -- (7.75,0);
		\draw[red,-] (7.25,0.5) -- (7.75,1);
		\draw[red,-] (7.25,1) -- (7.75,0.5);
		\draw[red,-] (7.75,0) -- (8.25,0.5);
		\draw[red,-] (7.75,0.5) -- (8.25,0);
		
		\draw[red,-] (7.75,0.5) -- (8.25,1);
		\draw[red,-] (7.75,1) -- (8.25,0.5);
		\draw[red,-] (7.75,1) -- (8.25,1.5);
		\draw[red,-] (7.75,1.5) -- (8.25,1);
		
		\draw (5.75,0) rectangle (8.25,0.5);
		\draw (5.75,0) rectangle (8.25,1);
		\draw (5.75,0) rectangle (6.25,1.5);
		\draw (5.75,0) rectangle (6.75,1.5);
		\draw (5.75,0) rectangle (7.25,1.5);
		\draw (5.75,0) rectangle (8.25,1.5);
		\draw[thick,-] (7.75,1.5) -- (8.25,1.5);
		\draw[thick,-] (7.75,.5) -- (7.75,1);
		\draw (5.75,1.5) rectangle (6.25,2);
		\draw (6.25,1.5) rectangle (6.75,2);
		\draw (6.75,1.5) rectangle (7.25,2);
		\draw (7.25,1.5) rectangle (7.75,2);
		\draw (7.75,1.5) rectangle (8.25,2);
		\draw
		(1.2,-0.2) node[anchor=north] {(A) Reliable Cells} ;
		\draw (3.8,-0.2) node[anchor=north] {(B) Defect} 
		(7,-0.2) node[anchor=north] {(C) Partially Defect};
		\end{tikzpicture}
	}
	\label{Fig1}
	\caption{Difference between reliable and (partially) defect memory cells. In this figure, there are $n=5$ cells with $q=4$ possible levels. The cell levels $\in \mathbb{F}_4$ are mapped to (0, 1, $\alpha$ or $1+\alpha$). Case (A) illustrates only reliable cells that can store any of the four values. In the stuck scenario, as shown in case (B), the defect cells can store only the exact stuck level $s$. Case (C) is more flexible (partially defect scenario). Partially stuck cells at level $s \geq 1$ can store level $s$ or higher. }	
\end{figure}
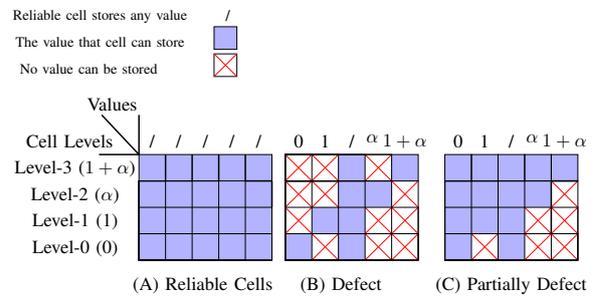
 
\subsection{Related Work}

In \cite{heegard1983partitioned}, code constructions for masking stuck memory cells were proposed. In addition to masking the stuck cells, it is possible to correct errors that occur during the storing and reading processes. A generator matrix of a specific form was constructed for this purpose. In \cite{wachterzeh2016codes}, improvements on the redundancy necessary for masking \emph{partially} stuck cells are achieved, and lower and upper bounds are derived.
However, the paper does not consider error correction in addition to masking.

Combined methods of \cite{heegard1983partitioned} and \cite{wachterzeh2016codes} to obtain code constructions for joint masking partially stuck cells and error correction are conducted in \cite{haideralkim2019psmc}. These code constructions reduce the redundancy necessary for masking, similar to the results in \cite{wachterzeh2016codes}.
In contrast to \cite{wachterzeh2016codes}, however, these constructions are able to correct additional random errors.
\subsection{Our Contribution}

In this paper, we extend the constructions of \cite{haideralkim2019psmc}. We obtain a code construction for combined error correction and masking $q$-ary partially defect cells by means of binary stuck memory cells. Our new construction gives higher rates for several sets of parameters compared to \cite[Theorem~4]{haideralkim2019psmc} and it can correct errors compared to \cite[Theorem~9]{wachterzeh2016codes}. In this paper, we also analyze how close our constructions are to the sphere packing bound (necessary condition) as an upper bound with the presence of partially defect memory constraints. Further, we introduce a Gilbert-Varshamov-like bounds (sufficient condition) to show the existence of codes with certain parameters that can mask the partially stuck memory and correct errors. 

Similar to the main part of \cite{wachterzeh2016codes}, this paper deals with partially-stuck-at-1 cells, i.e., $s=1$, but the results are extendable to arbitrary $s$ similar to \cite[Section~VII]{wachterzeh2016codes}.

\section{Preliminaries}

\subsection{Notations }\label{ssec:notation}
For a prime $q$, let $\mathbb{F}_q$ denote the finite field of order $q$ and ${\mathbb{F}}_q[x]$ be the set of all univariate polynomial with coefficients in $\mathbb{F}_q$. $\mathbb{F}_{q^\lambda}$ denotes an extension field of $\mathbb{F}_q$ of extension degree $\lambda$.
Denote $[f] = \{0,1, \dots, f-1\}$ for $f \in \mathbb{Z}_{>0}$. 

Throughout this paper, let $n$ be the total number of cells, $u$ be the number of (partially) stuck cells, and $t$ be the number of random errors.
Let $s_{\phi_i}$ denote the (partially) stuck level at position $i$, where $i \in [u]$, and $\ve{\phi}= \{\phi_0,\phi_1, \cdots,\phi_{u-1}\} \subseteq [n]$ denotes the positions of the the (partially) stuck cells. 

For our construction, let $k_1$ be the number of information symbols, $l$ be the number of symbols required for masking, and $r$ be the required redundancy for error correction.

\subsection{Definitions}
\subsubsection{Defect and Partially Defect Cells}
A cell is called a defect (\emph{stuck at level $s$}), where $ s \in [q]$, if it cannot change its value and always stores the value $s$. 
A cell is called partially defect (\emph{partially stuck at level $s$}), where $s \in [q]$, if it can only store values which are at least $s$.
We fix throughout the paper a total ordering ``$\geq$'' of the elements of $\Fq$ such that $a \geq 1 \geq 0$ for all $a \in \mathbb{F}_q^*$ (note that such an ordering does not interact well with addition, but this is not relevant here).
If a cell is partially stuck at 0, it is a non-defect cell which can store any of the $q$ levels.
\subsubsection{($u$, $t$)-PSMC}
An $(n,M)_q$  ($u$, $t$)-\emph{partially-stuck-at-masking code} $\mycode{C}$ is a coding scheme consisting of an encoder $\mathcal{E}$ and a decoder $\mathcal{D}$.
The input of the encoder $\mathcal{E}$ is
\begin{itemize}
\item the set of locations of $u$ partially stuck cells $\ve{\phi}= \{\phi_0,\phi_1, \dots,\phi_{u-1}\} \subseteq [n]$, 
\item the partially stuck levels $s_{\phi_0},s_{\phi_2}, \dots,s_{\phi_{u-1}}  \in [q]$, 
\item a message $\m \in \mathcal{M}$, where $\mathcal{M}$ is a message space of cardinality $|\mathcal{M}|$.
\end{itemize}
It outputs a vector $\c \in \Fq^n$ which fulfills ${c}_{\phi_i} \geq s_{\phi_i}$ for all $i=1,\dots,u$.
The decoder is a mapping that takes $\c+\e \in \mathbb{F}_q^n$ as input and returns the correct message $\m$ for all error vectors $\e$ of Hamming weight at most $t$.

\subsection{Code Construction over $\F_{2^\lambda}$}
The new code construction works over the finite field $\F_{2^{\lambda}}$. We denote by $x^0,x^1,\dots,x^{\lambda-1}$ a basis of $\F_{2^\lambda}$ over $\F_{2}$. Hence, any element $a \in \F_{2^\lambda}$ can be uniquely represented as $a=\sum_{i=0}^{\lambda-1} a_i x^i$ where $a_i \in \F_2 \forall i$.
In particular, $a \in \F_2$ if and only if $a_1=\dots=a_{\lambda-1}=0$. This is a crucial property of $\F_{2^\lambda}$ that we will use in Construction~\ref{cons2}.

In the construction, we also use the notation $(\ve{\Gamma})_{\ve{\phi}}$ for $\ve{\Gamma} \in \F_{2^\lambda}^{n+1}$ and $\ve{\phi} \subseteq \{0,\dots,n\}$, by which we mean the sub-vector of $\ve{\Gamma}$ indexed by the entries of $\ve{\phi}$.

\subsection{Error Models}
Assume that the memory has $u$ partially stuck cells at positions $\ve{\phi}= \{\phi_0,\phi_1, \dots,\phi_{u-1}\} \subseteq [n]$ and $n$ cells in total. 

In the \emph{non-overlapping model}, the $t$ errors can happen only at positions $\ve{\Psi}= \{\Psi_0,\Psi_1, \dots,\Psi_{{n-u}-1}\} = [n] \setminus \ve{\phi}$. 

In the \emph{overlapping model}, on the other hand, we assume that $t$ errors can happen in any cell, i.e., $\ve{\Psi} \subseteq [n]$. If errors happen in the $u$ partially stuck cells, we assume that the {error attains only values such that the corrupted vector still obeys the partially stuck constraints.}

\section{Codes for (Partially) Defect Memories}

We propose a new code construction for simultaneous masking and error correction.
The new construction is based on the masking-only construction in \cite[Section~\rom{6}]{wachterzeh2016codes}, which is able to mask $u\geq q$ partially stuck positions, where $q$ is the field size of the masking code, but cannot correct any errors.
We generalize this construction to be able to cope with errors.
Compared to \cite[Theorem~4]{haideralkim2019psmc}, the new construction may lead to larger code dimensions for a given pair ($u$, $t$), in a similar fashion as \cite[Construction~5]{wachterzeh2016codes} improves compared to \cite[Construction~4]{wachterzeh2016codes}.

\begin{const}\label{cons2}
Let $n,u,t,\lambda,k,k_1,r,l$ be positive integers with $u,t \leq n$, $\lambda > 1$, $k_1 = n - l - r -1$, and $k = l + k_1 +1$. 
Suppose that there are matrices
\begin{itemize}
\item $\P \in \F_{2^\lambda}^{k_1 \times r}$ and
\item $\H_0 := [\ve{I}_{l \times l} \mid \R_{l \times (n-l)}] \in \F_2^{l \times n}$ that is a systematic 
parity-check matrix of a binary code $\mycode{C}_0$ with parameters $[n, k_1+r,d_0 \geq u_0+1]_2$, where $u_0 := \lfloor 2u/2^\lambda\rfloor$,
\end{itemize}
such that
\begin{center}			
\scalebox{1}{	
$ \ve{G} = \begin{bmatrix} &\ve{H}_0&& 0 \\&\ve{G}_1&&\vdots \\&&&0 \\ 1&\dots&&1   \end{bmatrix}
:= \begin{bmatrix}  \ve{I}_{l \times l} & \multicolumn{2}{c}{\text{--- }\R_{l \times (n-l)}\text{ ---}}  & \ve{0}_{l\times1}   \\  \ve{0}_{k_1 \times l} & \ve{I}_{k_1 \times k_1} &\ve{P}_{k_1 \times r} &\ve{0}_{k_1\times1} \\\multicolumn{4}{c}{\text{-------------- }\ve{1}_{1 \times (n+1)}\text{ --------------}}  \end{bmatrix}$}
\end{center}
is a generator matrix of a $[n+1, k = l+k_1+1, d \geq 2t+1]_{2^\lambda}$ code $\Code$. 

Based on these definitions, we define a coding scheme in Algorithm~\ref{a13} and Algorithm~\ref{a14}.   
\end{const}

\begin{algorithm}[ht]
	
	\caption{Encoding ($\ve{m}; \ve{m^\prime}; \ve{\phi}$)}
	\label{a13}
	\KwIn{
		\begin{itemize}
			\item Message: \\
			$\ve{m} = (m_0, m_1,\dots, m_{k_1-1})  \in {\mathbb{F}_{2^\lambda}^{k_1}}$ and\\ \vspace{0.05 cm} 
			$\ve{m^\prime} = (m^\prime_0, m^\prime_1,\dots, m^\prime_{l-1}) \in \mathcal{F}^{l}$, where $\mathcal{F} := \{ \sum_{i=0}^{\lambda-1} a_i x^i \in \F_{2^\lambda} \, : \, a_0=0 \} \subseteq \F_{2^\lambda}$.
			\item Positions of partially stuck-at-($s=1$) cells: $\ve{\phi} \subseteq \{0,\dots,n\}$
			\item Notions introduced in Construction~\ref{cons2}.
		\end{itemize}
	}
	$\a \gets \m' \cdot \big[\H_0 \mid 0\big]$ \\
	$\b \gets \m \cdot [\G_1 \mid 0]$ \\ 
	$\ve{w} \gets \a + \b + (z+1) \cdot \ve{1}_{1\times (n+1)}$, where $z \in \mathcal{F}$ is chosen such that $(\a + \b + (z+1) \cdot \ve{1}_{1\times (n+1)})_{\ve{\phi}}$ has at most $u_0 = \lfloor 2^{1-\lambda} u \rfloor$ entries in $\F_2$ \\
	Choose $\ve{z} \in \F_2^{l}$ such that, for any $i=0,\dots,u-1$,
	\begin{equation*}
	(\ve{z} \cdot \H_0 \mid 0)_{\phi_i} =
	\begin{cases}
	0, &\text{if $(\ve{w})_{\phi_i}=1$}, \\
	1, &\text{if $(\ve{w})_{\phi_i}=0$}, \\
	\text{arbitrary}, &\text{if $(\ve{w})_{\phi_i} \notin \F_{2}$}. \\
	\end{cases}
	\end{equation*} \\
	\KwOut{$\ve{w}+[\ve{z} \cdot \H_0 \mid 0]$}

\end{algorithm}

\begin{algorithm}[ht]
	\label{a14}
	\caption{Decoding}	
	\KwIn{
		\begin{itemize}
			\item $\ve{y} = \ve{c}+\ve{e} \in \mathbb{F}^{n+1}_{2^\lambda}$, where $\ve{c}$ is a valid output of Algorithm~\ref{a13} and $\ve{e}$ is an error of Hamming weight at most $t$.
			\item Notions introduced in Construction~\ref{cons2}.
		\end{itemize}
	}
	$\v \gets$ decode $\ve{y}$ in the code $\mathcal{C}$ \\
	$\v' \gets \ve{v}-(v_n+1)\ve{1}_{1 \times (n+1)}$ \\
	$\hat{\m}'' \gets [v'_0,\dots,v'_{l-1}]$ \\
	$\hat{\m}' \gets [\varphi(m''_0),\dots,\varphi(m''_{l-1})]$, where
	\begin{equation*}
\varphi \, : \, \F_{2^\lambda} \to \mathcal{F}, \quad 
\sum_{i=0}^{\lambda-1} a_i x^i \mapsto 0 x^0 + \sum_{i=1}^{\lambda-1} a_i x^i.
\end{equation*} \\
	$\ve{v}'' \gets \ve{v}'- \hat{\m}'' \cdot \begin{bmatrix}  \ve{I}_{l \times l} & \R_{l \times (n-l)}  & \ve{0}_{l\times1} \end{bmatrix}$ \\
	$\hat{\m} \gets [v_{l}'',\dots,v_{l+k_1-1}]$ \\
	\KwOut{$\hat{\m}$ and $\hat{\m}'$}
	
\end{algorithm}	

\begin{thmmystyle}
\label{thm9}
The coding scheme in Construction~\ref{cons2} is a ($u,t$)-PSMC of length $n+1$ and cardinality
\begin{equation*}
\mathcal{M}_{u,t} = 2^{\lambda (k_1+l)-l}.
\end{equation*}
\end{thmmystyle}

\begin{proof}
\textbf{Masking:}
We first prove that Algorithm~\ref{a13} outputs a masked vector (i.e., a vector that contains no zeros in the partially stuck-at-($s=1$) positions $(\phi_0,\dots,\phi_{u-1})$).

By construction, the vector $(\a + \b)_{\ve{\phi}}$ has length $u$ and its entries are in $\F_{2^\lambda}$.
Consider the partition $\{c,c+1\}$, for all $c \in \mathcal{F}$, of $\F_{2^\lambda}$. These sets are pairwise disjoint for different $c$ and there are $2^{\lambda-1}$ such sets.
By the pigeonhole principle, there is one such set, say $\{-z,-z+1\}$, such that the vector $(\a + \b + (z+1) \cdot \ve{1}_{1\times (n+1)})_{\ve{\phi}}$ contains at most $\lfloor \tfrac{u}{|\mathcal{F}|} \rfloor = \lfloor 2^{1-\lambda} u \rfloor = u_0$ elements in $\{0,1\}$.

We need to mask the few remaining entries of $\ve{w}$ that are $0$.
First note that since the last entry of $\a+\b$ is zero and we add $z+1 \neq 0$ to it, the last entry of $\ve{w}$ is always non-zero (this is relevant if the last position is partially stuck).
We choose a binary vector $\ve{z}$ such that $[\ve{z} \cdot \H_0 \mid 0]_\ve{\phi}$ contains a $1$ (or $0$) in positions in which the corresponding entry of $(\ve{w})_\ve{\phi}$ is $0$ (or $1$, respectively).
Such a vector $\ve{z}$ always exists since
\begin{itemize}
\item the number of $\{0,1\}$-entries in $(\ve{w})_\ve{\phi}$ is at most $u_0$,
\item $\H_0$ is a parity-check matrix of a binary code of minimum distance $\geq u_0+1$, which means that any $u_0$ columns of $\H_0$ are linearly independent, and
\item the vector $[\ve{z} \cdot \H_0 \mid 0]$ is in $\F_2$ (since $\ve{z}$ and $\H_0$ are binary).
\end{itemize}
Hence, all entries of $(\ve{w}+[\ve{z} \cdot \H_0 \mid 0])_\ve{\phi}$ are either $1$ or in $\F_{2^\lambda}\setminus \F_2$, i.e., non-zero.

\textbf{Error correction:}
Next, we show that the output of Algorithm~\ref{a13} is a codeword of the code $\mathcal{C}$ as defined in Construction~\ref{cons2}, and hence, the first step of Algorithm~\ref{a14} is able to correct up to $t$ errors in the masked vector.

To prove this, we rewrite
\begin{align}
&\ve{w}+[\ve{z} \cdot \H_0 \mid 0] \notag \\
&= \a + \b + (z+1) \cdot \ve{1}_{1\times (n+1)} + [\ve{z} \cdot \H_0 \mid 0] \notag \\
&= \big[ \m'+\ve{z} \mid \m \mid z+1 \big] \cdot \begin{bmatrix} &\ve{H}_0&& 0 \\&\ve{G}_1&&\vdots \\&&&0 \\ 1&\dots&&1   \end{bmatrix} \in \mathcal{C} \label{eq:masked_codeword}
\end{align}

\textbf{Recovery of the messages:}
We show that Algorithm~\ref{a14} uniquely retrieves the messages $\m$ and $\m'$ from the vector $\ve{v} := \ve{w}+[\ve{z} \cdot \H_0 \mid 0]$, which is obtained after the error-correction step.

By \eqref{eq:masked_codeword}, we have
\begin{align*}
\ve{v} = \big[ \m'+\ve{z} \mid \m \mid z+1 \big] \begin{bmatrix}  \ve{I}_{l \times l} & \multicolumn{2}{c}{\text{--- }\R_{l \times (n-l)}\text{ ---}}  & \ve{0}_{l\times1}   \\  \ve{0}_{k_1 \times l} & \ve{I}_{k_1 \times k_1} &\ve{P}_{k_1 \times r} &\ve{0}_{k_1\times1} \\\multicolumn{4}{c}{\text{-------------- }\ve{1}_{1 \times (n+1)}\text{ --------------}}  \end{bmatrix}
\end{align*}
Hence, we can retrieve $z+1$ as the last entry of $\ve{v}$ and subtract it from $\ve{v}$, i.e.,
\begin{align*}
\ve{v}' :=& \; \ve{v}-(z+1)\ve{1}_{1 \times (n+1)} \\
=& \; \big[ \m'+\ve{z} \mid \m\big] \begin{bmatrix}  \ve{I}_{l \times l} & \multicolumn{2}{c}{\text{--- }\R_{l \times (n-l)}\text{ ---}}  & \ve{0}_{l\times1}   \\  \ve{0}_{k_1 \times l} & \ve{I}_{k_1 \times k_1} &\ve{P}_{k_1 \times r} &\ve{0}_{k_1\times1}  \end{bmatrix}
\end{align*}
Then, the first $l$ positions of $\ve{v}'$ equals $\hat{\m}'' = \m'+\ve{z}$. Thus, we can also subtract
\begin{align*}
\ve{v}'' :=& \; \ve{v}'- (\m'+\ve{z}) \cdot \begin{bmatrix}  \ve{I}_{l \times l} & \R_{l \times (n-l)}  & \ve{0}_{l\times1} \end{bmatrix}\\
=& \; \m \cdot  \begin{bmatrix}  \ve{0}_{k_1 \times l} & \ve{I}_{k_1 \times k_1} &\ve{P}_{k_1 \times r} &\ve{0}_{k_1\times1}  \end{bmatrix}
\end{align*}
and obtain $\m = \hat{\m}$ as the $l+1$st to $l+k_1$-th position of $\ve{v}''$.

It remains to prove that we can uniquely retrieve $\m' = \hat{\m}'$ from $\m'+\ve{z} = \hat{\m}''$.
Recall that the entries of $\m'$ are in $\mathcal{F}$ (i.e., elements are, represented in a polynomial basis $x^i$ of $\F_{2^\lambda}$ over $\F_2$, of the form $\sum_{i=0}^{\lambda-1} a_i x^i$ with $a_0=0$).
Hence, for $\alpha \in \F_2$ and $\beta = 0 x^0 + \sum_{i=1}^{\lambda-1} a_i x^i \in \mathcal{F}$, we have
\begin{align*}
\varphi(\alpha+\beta) &= \varphi\!\left(\alpha x^0 + \sum_{i=1}^{\lambda-1} a_i x^i \in \mathcal{F} \right) \\
&= 0 x^0 + \sum_{i=1}^{\lambda-1} a_i x^i \in \mathcal{F} = \beta.
\end{align*}

In summary, Algorithm~\ref{a13} encodes messages $\m$ and $\m'$ into a codeword of the code $\mathcal{C}$ such that all $u$ partially-stuck positions are non-zero (i.e., masked).
Algorithm~\ref{a14} then corrects up to $t$ errors using the code $\mathcal{C}$ and recovers the message vectors.
This means that the coding scheme is a ($u$, $t$)-PSMC.
The cardinality of the code is given by
\begin{align*}
\mathcal{M}_{u,t} = |\F_{2^\lambda}|^{k_1} |\mathcal{F}|^l = 2^{\lambda (k_1+l)-l},
\end{align*}
which concludes the proof.
\end{proof}

On the first glance, it is not immediately clear how to construct the matrices $\H_0$ and $\G_1$ in Theorem~\ref{thm9}. An intuition is as follows: we look for a code $\Code$ that 
\begin{itemize}
\item has minimum distance $\geq 2t+1$ (to correct at least $t$ errors),
\item contains the all-one vector, and
\item when punctured at the last position, its binary subfield subcode must contain a code (generated by $\H_0$) whose dual code has minimum distance at least $u_0+1$.
\end{itemize}

Construction~\ref{cons2} can give higher rates for some chosen codes than using \cite[Theorem~4]{haideralkim2019psmc}. Example~1 in the appendix shows a code $\mathcal{C}$ of parameters $[15,11,3]_4$ and its generator matrix $\ve{G}$, $u=4$ and $t=1$. Using Theorem~\ref{thm9}, the cardinality $\mathcal{M}_{u,t} = 2^{2 (6+4)-4}= 4^{8}$ ($k_1=8$ information symbols). In contrast \cite[Theorem~4]{haideralkim2019psmc}, the cardinality $\mathcal{M}_{u,t} = 4^{7}$ ($k_1=7$ information symbols) using $\ve{G}'$, while in both theorems we can mask $u=q$ cells.\\

\begin{figure*}[b]
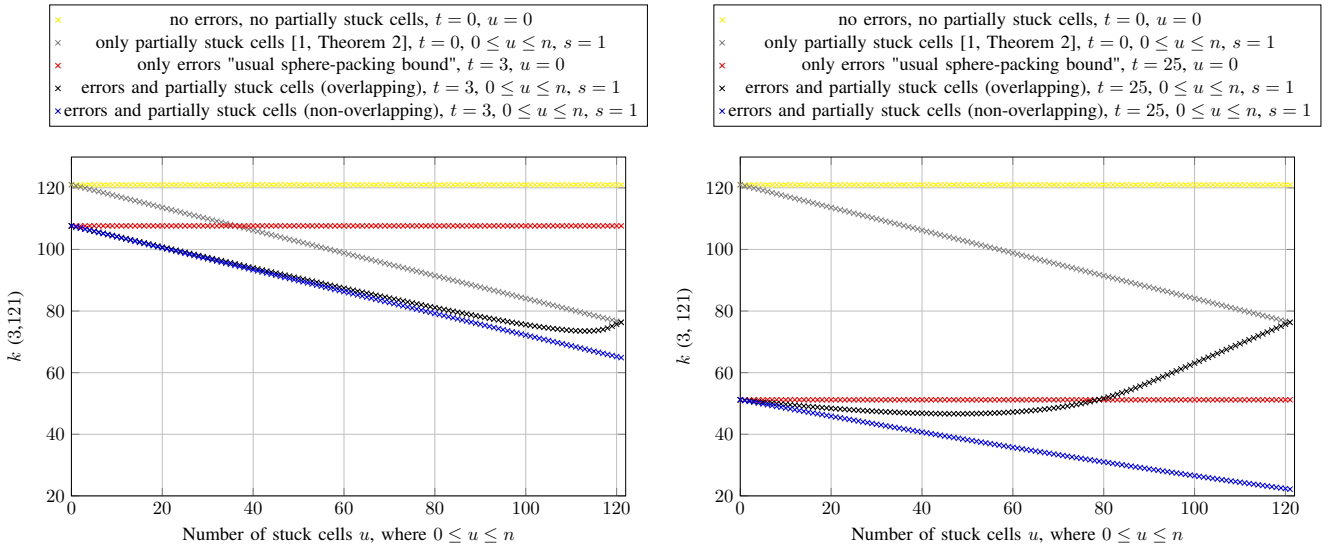

	\scalebox{0.70}{
	\begin{tikzpicture} 
	\begin{axis}[
	filter discard warning=false,
	height=8cm,
	width=12cm,
	ylabel= {$k$ ($3$,$121$)},
	xmin = 0,
	xmax = 122,
	ymin = 20,
	ymax = 130,
	xlabel= {Number of stuck cells $u$, where $0\leq u\leq n$},
	grid=major,
	legend style={at={(0.5,1.45)},anchor=north}
	]
	\def\mymark{x}
	\input{figures/SPq3n121t3/k_Only_Normal_Cells}
	\addlegendentry{no errors, no partially stuck cells, $t=0$, $u=0$};
	\input{figures/SPq3n121t3/k_cases_Only_PSMC}
	\addlegendentry{only partially stuck cells \cite[Theorem~2]{wachterzeh2016codes}, $t=0$, $0\leq u \leq n$, $s=1$};
	\input{figures/SPq3n121t3/k_cases_sphere_packing}
	\addlegendentry{only errors "usual sphere-packing bound", $t =3$, $u=0$};
	\input{figures/SPq3n121t3/k_cases_non_overlapping}
	\addlegendentry{errors and partially stuck cells (overlapping), $t=3$, $0\leq u \leq n$, $s=1$};
	\input{figures/SPq3n121t3/k_cases_overlapping}
	\addlegendentry{errors and partially stuck cells (non-overlapping), $t=3$, $0\leq u \leq n$, $s=1$};
	
	\end{axis}
	
	\end{tikzpicture} }
	\scalebox{0.70}{
	\begin{tikzpicture} 
	\begin{axis}[
	filter discard warning=false,
	height=8cm,
	width=12cm,
	ylabel= {$k$ ($3$, $121$)},
	xmin = 0,
	xmax = 122,
	ymin = 20,
	ymax = 130,
	xlabel= {Number of stuck cells $u$, where $0\leq u\leq n$},
	grid=major,
	legend style={at={(0.50,1.45)},anchor=north}
	]
	\def\mymark{x}
	\input{figures/SPq3n121t25/k_Only_Normal_Cells}
	\addlegendentry{no errors, no partially stuck cells, $t=0$, $u=0$};
	\input{figures/SPq3n121t25/k_cases_Only_PSMC}
	\addlegendentry{only partially stuck cells \cite[Theorem~2]{wachterzeh2016codes}, $t=0$, $0\leq u \leq n$, $s=1$};
	\input{figures/SPq3n121t25/k_cases_sphere_packing}
	\addlegendentry{only errors "usual sphere-packing bound", $t =25$, $u=0$};
	\input{figures/SPq3n121t25/k_cases_non_overlapping}
	\addlegendentry{errors and partially stuck cells (overlapping), $t=25$, $0\leq u \leq n$, $s=1$};
	\input{figures/SPq3n121t25/k_cases_overlapping}
	\addlegendentry{errors and partially stuck cells (non-overlapping), $t=25$, $0\leq u \leq n$, $s=1$};
	
	\end{axis}
	
	\end{tikzpicture}}
	\caption{\textbf{Sphere-packing bounds}: Comparison for $k (q,n)$ information symbols for the classical sphere-packing bound ("only errors") and our sphere-packing-like bounds ("errors and stuck cells") for non-overlapping and overlapping errors.
		The chosen parameters are $\lambda =5$ and $ q =3$, and $n = ((q^\lambda -1)/(q-1))$.}
	\label{fig2}
\end{figure*}

\section {Bounds on the Cardinality and Minimum Distance  }

In this section, we derive bounds on ($u$,$t$)-PSMCs with the goal to evaluate the parameters of code constructions.
We derive a sphere-packing-type bound (necessary condition) and a \emph{Gilbert-Varshamov} (GV)-type bound (sufficient condition). Notice that the GV-like bound proves the existence of ($u$,$t$)-PSMCs only when $u <q$.
\subsection{Sphere-Packing Bound on PSMCs}
The \emph{sphere-packing bound} considers the size of a code by packing spheres around each codeword. We combine the proof of the classical sphere packing bound with constraints from partially stuck cells. 
 \begin{thmmystyle}
	\label{thm7}
	Any ($u$,$t$)-PSMC over $\mathbb{F}_q$ of cardinality $\mathcal{M}_{u,t}$ has to satisfy:
	\begin{itemize}
		\item for non-overlapping errors: 
		\begin{equation}\label{eq8}
		\mathcal{M}_{u,t} \cdot \sum_{j=0}^{t} {n-u \choose j} (q-1)^{j} \leq q^{n-u} \prod_{i=0}^{u-1} (q-s_{\phi_i}),
		\end{equation}
        \item for overlapping errors under the assumption that $ s_{\phi_i}< q-1$, for $i=0,\dots,u-1$: 
	\begin{align}	
		 \hspace{-4ex}\mathcal{M}_{u,t} &\cdot \sum_{j=0}^{t}\sum_{j_1=0}^{j}\tbinom{n-u}{j_1}
		 (q-1)^{j_1}  \tbinom{u}{j-j_1}\prod_{i\in \mathcal{J}}
		 (q-1-s_{\phi_i})\nonumber\\ 
		 &\leq q^{n-u} \prod_{i=0}^{u-1} (q-s_{\phi_i}),\label{eq9}
	\end{align}
	where $\mathcal{J}$ denotes the set of cardinality $j-j_1$ of stuck cells that is affected by errors.
	\end{itemize}
 
\end{thmmystyle} 
\begin{proof}
	First we prove the \textbf{non-overlapping} scenario. \cite[Theorem~2]{wachterzeh2016codes} proves that at most $\mathcal{M}_u $ $q$-ary codewords can be stored in a memory in the presence of $u$ partially stuck cells, where:
	\begin{equation} \label{eq10}
	\mathcal{M}_u  \leq q^{n-u} \prod_{i=1}^{u} (q-s_{\phi_i}).
	\end{equation} 
    We assume that the errors can happen only in the $n-u$ non-stuck cells. A sphere $\mathcal{B}_{t,n-u}(\vec{a})$ of radius $t$ around a word $\vec{a} \in \mathbb{F}^{n-u}_q$ is the set of all words in Hamming distance at most $t$, i.e, $\mathcal{B}_{t,n-u}(\vec{a}) := \{ \vec{b}\in \mathbb{F}^{n-u}_q: d(\vec{a},\vec{b})\leq t\}$.
    There are $n-u \choose j$ words in distance exactly $j$ from a fixed word and $|\mathcal{B}_{t,n-u}(\vec{a})|= \sum_{j=0}^{t} {n-u \choose j} (q-1)^{j}$.    
	Thus, the total number of words in all decoding spheres (left-hand side (LHS) of~\eqref{eq8})
	is at most the total number of possible words $\mathcal{M}_u$ (right-hand side (RHS) of~\eqref{eq8}).
	
 Second we prove the \textbf{overlapping} scenario. 
 Clearly, \eqref{eq10} is still an upper bound on the total number of possible words, i.e., the RHS of the sphere-packing bound.
 For the LHS in this case, the errors can happen either in the $n-u$ non-stuck cells or in the $u$ stuck cells such that $s_{\phi_i} + e_i \leq q-1$.

 In this case, there are $\sum_{j_1=0}^{j}{n-u \choose j_1} {u \choose j-j_1}$ possibilities for $j$ erroneous positions (i.e., $j_1$ errors happen at non-stuck positions and $j-j_1$ errors happen at stuck positions). Therefore, there are $\mathcal{B}_{t,u,n}:=\sum_{j=0}^{t}\sum_{j_1=0}^{j}{n-u \choose j_1} (q-1)^{j_1}  {u \choose j-j_1} \prod_{i \in \mathcal{J}}(q-1-s_{\phi_i})$ distinct words that can result from a fixed word when $u$ partially stuck cells and at most $t$ random errors happen, where $\mathcal{J}$ denotes the set of cardinality $j-j_1$ of stuck cells that is affected by errors.

 Since the set of these $\mathcal{B}_{t,u,n}$ words around a fixed codeword is disjoint to the corresponding set around another fixed codeword, ($ \mathcal{M}_{u,t} \cdot\mathcal{B}_{t,u,n}$) is at most the total number of possibilities and the statement follows.
\end{proof}
Figure~\ref{fig2} illustrates the new sphere packing bounds. They are compared to the amount of storable information symbols for a completely reliable memory (i.e., no stuck cells, no errors) and the upper bound on the cardinality of an only-masking PSMC (only stuck cells, no errors) derived in \cite{wachterzeh2016codes}.
The figure also compares the overlapping and non-overlapping error model for $s_{\phi_i}= 1$, for all $i$. In the overlapping scenario, when the number of errors is small (e.g., $t=3$), the number of information symbols is bounded by the sphere packing bound and it is slightly better than in non-overlapping case. The more errors happen (e.g., $t=25$ in the right figure), a smaller number $(n-u)$ of non-stuck cells can be affected by errors, so it is very likely that many of the errors happen at partially stuck positions which affects the amount of storable information less than an error in a non-stuck position.

\subsection{Gilbert--Varshamov Bound on PSMCs}
 
In this section, we derive a sufficient condition for the existence of a code with certain parameters that can mask partially stuck cells and correct errors.
To derive this bound, we rely on our previous construction of PSMCs in \cite{haideralkim2019psmc}, which for $u<q$ masked cells solely required the existence of an error-correcting code with minimum distance $\geq 2t+1$, which contains the all-one vector.
In the proof, we therefore prove the existence of a code which contains the all-one vector as codeword.

\begin{thmmystyle} [Gilbert-Varshamov-like bound]
	\label{th11} 
	Let the positive integers $n$, $k \leq n$, $d \leq n$, $q$ fulfill:
	\begin{equation}\label{eq:GV-like-bound}
	\sum_{i=0}^{d-2} \binom{n -1}{i} (q-1)^i
	< q^{n-k}.
	\end{equation}
	Then, there exists an $[n',k',d]_q$ code that contains the all-one vector, where $n^\prime$, and $k^\prime$ satisfy: 
    \begin{equation*}
n-d+2 \leq n' \leq n+1,
\qquad
	 k-d+2 \leq k^\prime \leq k+1.
	\end{equation*}
	The parity-check matrix of this $[n',k',d]_q$ code can be constructed as shown in the proof.
\end{thmmystyle}
\begin{proof}
	Similar to the proof of the standard Gilbert--Varshamov bound, we construct a systematic parity-check matrix by adding columns $\ve{h}_l$ for $l=k+1,k+2,...$ to a $k \times k$ identity matrix as long as:
			\begin{equation} \label{eq16}
			 \sum_{i=0}^{d-2} \binom{l-1}{i} \cdot (q-1)^{i} < q^{n-k}.
			 \end{equation}
	Recall from the proof of the Gilbert-Varshamov bound that this condition ensures that there exists a column $\vec{h}_l$ that is linearly independent of any collection of $d-2$ other columns.
	
	If~\eqref{eq16} is not fulfilled anymore for $l=n+1$, we append an additional parity-check column $\ve{p}$ to the previous $n$ columns such that the sum of each row is zero (i.e., the weight is even in the binary case). This matrix is therefore:	
			\[
			\ve{H}_{e} := 
			\left[\begin{array}{@{}c|c}
			\begin{matrix}
			 & \Big{(} \underbrace{\ve{h}_1, \dots, \ve{h}_{n}}_{n} \Big{)}\\ 
			\end{matrix}
			& \ve{p}
			\end{array}\right],
			\]
where
	\begin{equation}\label{eq17}
		\sum_{i=1}^{n} \ve{h}_{i}+\ve{p}  = \ve{0}.
	\end{equation}  
However, for $\ve{H}_{e}$, we cannot guarantee anymore that any $d-1$ columns are linearly independent (as $\ve{p}$ might be linearly dependent on a small number of $\ve{h}_i$'s.). Therefore, in the following, we possibly remove a few columns from $\ve{H}_{e}$ to recover this property while still having zero row sums.

If $\ve{p}$ is linearly independent of any $d-1$ columns in $\vec{h}_1, \dots, \vec{h}_n$, we define $\ve{H}:=\ve{H}_{e}$.

	Else $\ve{p}$ is linearly dependent of $\delta \leq d-2$ columns $\{\vec{h}_{i_1},\vec{h}_{i_2},\dots, \vec{h}_{i_{\delta}}\} \subseteq \{\vec{h}_1,\vec{h}_2,\dots, \vec{h}_{n}\}$, and $\ve{p}$ is a linear combination of these $\delta$ columns:
	\begin{equation*}
		\ve{p} = \sum_{j=1}^{\delta} \ve{h}_{i_j} \cdot a_{j}, \text{ where } a_{j} \in \{1,2,\dots, q-1\}.
	\end{equation*}
	Thus with $1 \leq \delta_1 \leq\dots\leq \delta_{q-1}\leq \delta$ (by assuming w.l.o.g. an ordering on the indices),	
	\begin{equation}\label{eq18}
	\ve{p} = \sum_{j=1}^{\delta_1} \ve{h}_{i_j} + 2\cdot \sum_{j=\delta_1+1}^{\delta_2} \ve{h}_{i_j} +\dots+ (q-1)\cdot \sum_{j=\delta_{q-1}+1}^{\delta} \ve{h}_{i_j}.
	\end{equation}
	
	We can rewrite \eqref{eq17} as:
		\begin{equation*}
	\hspace{-3ex}\sum_{i=1 \setminus \{i_1,\dots,i_{\delta}\} }^{n} \hspace{-3ex}\ve{h}_{i}+
	\sum_{j=1}^{\delta_1} \ve{h}_{i_j} + \sum_{j=\delta_1+1}^{\delta_2} \ve{h}_{i_j} +\dots+ \sum_{j=\delta_{q-1}+1}^{\delta} \ve{h}_{i_j}
	+ \ve{p}  = \ve{0}.
	\end{equation*}

Combining this with~\eqref{eq18} yields:
     \begin{align*}
&\hspace{-2.5ex}\sum_{i=1 \setminus \{i_1,\dots,i_{\delta}\} }^{n} \hspace{-3ex}\ve{h}_{i}+
\sum_{j=1}^{\delta_1} \ve{h}_{i_j} + \sum_{j=\delta_1+1}^{\delta_2} \ve{h}_{i_j} +\dots+ \sum_{j=\delta_{q-1}+1}^{\delta} \ve{h}_{i_j}
+ \ve{p}  \\
&\hspace{-1.5ex}+\sum_{j=1}^{\delta_1} \ve{h}_{i_j} + 2\cdot \sum_{j=\delta_1+1}^{\delta_2} \ve{h}_{i_j} +\dots+ (q-1)\cdot \sum_{j=\delta_{q-1}+1}^{\delta} \ve{h}_{i_j} - \ve{p}.\\
&=\ve{0}.
     \end{align*}
Therefore,
\begin{align}
&\hspace{-2ex}\sum_{i=1 \setminus \{i_1,\dots,i_{\delta}\} }^{n} \hspace{-3ex}\ve{h}_{i}+
(2 \mod q)\sum_{j=1}^{\delta_1} \ve{h}_{i_j} + (3 \mod q)\sum_{j=\delta_1+1}^{\delta_2} \ve{h}_{i_j}\nonumber\\ &+\dots+ (q-1 \mod q)\sum_{j=\delta_{q-2}+1}^{\delta_{q-1}} \ve{h}_{i_j}=\ve{0}.\label{eq:GV-stillevenweight}
\end{align}

Therefore, the matrix
\begin{align*}
&\vec{H}:=\\
&\hspace{-2ex}\left(\begin{smallmatrix}
 \underbrace{\ve{h}'_1 ,\dots ,\ve{h}'_{n-\delta}}_{n-\delta}
&  \underbrace{\Big | 2 \vec{h}_{i_1}, ..., 2 \vec{h}_{i_{\delta_1}}}_{\delta_1}
& \Big | \dots 
&  \underbrace{\Big | -\vec{h}_{i_{\delta_{q-2}+1}}, ..., - \vec{h}_{i_{\delta_{q-1}}}}_{\delta_{q-1}-\delta_{q-2}}
\end{smallmatrix}\right)
\end{align*}
where $\ve{h}'_1, \dots, \ve{h}'_{n-\delta} = \{\ve{h}_1,\dots,\ve{h}_n \} \setminus \{\ve{h}_{i_1},\dots,\ve{h}_{i_\delta} \}$, has sum equal to zero in all rows due to~\eqref{eq:GV-stillevenweight} and any $d-1$ columns are linearly independent since they are all columns (times a non-zero scalar) of the matrix $(\ve{h}_1,\dots,\ve{h}_n)$.

The number of columns $n'$ of $\ve{H}$ is bounded by
\begin{equation*}
n-d+2 \leq n-\delta \leq n' \leq n+1,
\end{equation*}
where $n'=n+1$ if $\ve{p}$ was linearly independent of any $d-2$ other columns and therefore no columns have to be removed.

Substituting $n$ in $l$ of~\eqref{eq16}, we obtain:

	\begin{equation}\label{eq19}
	\sum_{i=0}^{d-2} \binom{n -1}{i} (q-1)^i
	< q^{n-k}.
	\end{equation}
Since $n \leq n'+d-2$ and since $n-k=n'-k'$ (the number of rows did not change), we get 
	\begin{equation*}
\sum_{i=0}^{d-2} \binom{n' +d-3}{i} (q-1)^i
< q^{n'-k'}.
\end{equation*}
Thus, if this is true, there exists an $[n',k',d]_q$ code that contains the all-one vector, where $k'=n'-(n'-k')=n'-(n-k)=n'-n+k \geq n-d+2-n+k = k-d+2$.
\end{proof}
\begin{figure*} [h]
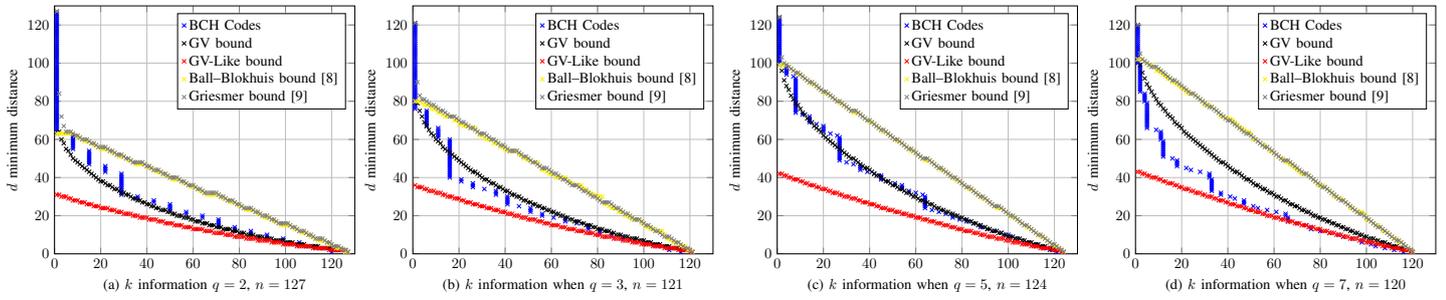

	\scalebox{0.55}{
		
		\begin{tikzpicture}
			\pgfplotsset{compat = 1.3}
			\begin{axis}[
				width = \columnwidth,
				xlabel = {(a) $k $ information $q =2$, $n=127$},
				ylabel = {$d$ minimum distance},
				xmin = 0,
				xmax = 130.0,
				ymin = 0,
				ymax = 130.0,
				grid=major,
				legend pos = north east,
				legend cell align=left]
				\def\mymark{x}
				\input{figures/figuresq2_fig_added/plot_different_k_for_BCH_Codes}
				\addlegendentry{BCH Codes};

				\input{figures/figuresq2_fig_added/plot_list_gilbert_varshamov_bound}
				\addlegendentry{GV bound};
				\input{figures/figuresq2_fig_added/plot_list_gilbert_varshamov_bound_fig}
				\addlegendentry{ GV-Like bound};
				\input{figures/figuresq2_fig_added/plot_list_ball_blokhuis_bound}
				\addlegendentry{Ball--Blokhuis bound \cite{ball2013bound}};
				\input{figures/figuresq2_fig_added/plot_list_griesmer_bound}
				\addlegendentry{Griesmer bound \cite{griesmer1960bound}};
			\end{axis}
		\end{tikzpicture}
		\begin{tikzpicture}
			\pgfplotsset{compat = 1.3}
			\begin{axis}[
				width = \columnwidth,
				xlabel = {(b) $k $ information when $q =3$, $n=121$},
				ylabel = {$d$ minimum distance},
				xmin = 0,
				xmax = 130.0,
				ymin = 0,
				ymax = 130.0,
				grid=major,
				legend pos = north east,
				legend cell align=left]
				\def\mymark{x}
				\input{figures/figuresq3_fig_added/plot_different_k_for_BCH_Codes}
				\addlegendentry{BCH Codes}
				\input{figures/figuresq3_fig_added/plot_list_gilbert_varshamov_bound}
				\addlegendentry{GV bound};
				\input{figures/figuresq3_fig_added/plot_list_gilbert_varshamov_bound_fig}
				\addlegendentry{GV-Like bound};
				\input{figures/figuresq3_fig_added/plot_list_ball_blokhuis_bound}
				\addlegendentry{Ball--Blokhuis bound \cite{ball2013bound}};
				\input{figures/figuresq3_fig_added/plot_list_griesmer_bound}
				\addlegendentry{Griesmer bound \cite{griesmer1960bound}};
			\end{axis}
	\end{tikzpicture}}
	\scalebox{0.55}{
		\begin{tikzpicture}
			\pgfplotsset{compat = 1.3}
			\begin{axis}[
				width = \columnwidth,
				xlabel = {(c) $k $ information when $q =5$, $n=124$},
				ylabel = {$d$ minimum distance},
				xmin = 0,
				xmax = 130.0,
				ymin = 0,
				ymax = 130.0,
				grid=major,
				legend pos = north east,
				legend cell align=left]
				\def\mymark{x}
				\input{figures/figuresq5_fig_added/plot_different_k_for_BCH_Codes}
				\addlegendentry{BCH Codes};

				\input{figures/figuresq5_fig_added/plot_list_gilbert_varshamov_bound}
				\addlegendentry{GV bound};
				\input{figures/figuresq5_fig_added/plot_list_gilbert_varshamov_bound_fig}
				\addlegendentry{GV-Like bound};
				\input{figures/figuresq5_fig_added/plot_list_ball_blokhuis_bound}
				\addlegendentry{Ball--Blokhuis bound \cite{ball2013bound}};
				\input{figures/figuresq5_fig_added/plot_list_griesmer_bound}
				\addlegendentry{Griesmer bound \cite{griesmer1960bound}};
			\end{axis}
		\end{tikzpicture}
		
		\begin{tikzpicture}
			\pgfplotsset{compat = 1.3}
			\begin{axis}[
				width = \columnwidth,
				xlabel = {(d) $k $ information when $q =7$, $n=120$},
				ylabel = {$d$ minimum distance},
				xmin = 0,
				xmax = 130.0,
				ymin = 0,
				ymax = 130.0,
				grid=major,
				legend pos = north east,
				legend cell align=left]
				
				\def\mymark{x}
				\input{figures/figuresq7_fig_added/plot_different_k_for_BCH_Codes}
				\addlegendentry{BCH Codes};

				\input{figures/figuresq7_fig_added/plot_list_gilbert_varshamov_bound}
				\addlegendentry{GV bound};
				\input{figures/figuresq7_fig_added/plot_list_gilbert_varshamov_bound_fig}
				\addlegendentry{GV-Like bound};
				\input{figures/figuresq7_fig_added/plot_list_ball_blokhuis_bound}
				\addlegendentry{Ball--Blokhuis bound \cite{ball2013bound}};
				\input{figures/figuresq7_fig_added/plot_list_griesmer_bound}
				\addlegendentry{Griesmer bound \cite{griesmer1960bound}};
			\end{axis}
	\end{tikzpicture}}
	
	\caption{Comparisons of bounds of our new GV-Like-Bound for different $n$ and $q$. For BCH codes, we use our previous construction (for $u<q$)\cite[Theorem~1]{haideralkim2019psmc} since BCH codes have the all-one vector in their generator matrices.} 
	\label{fig3}
\end{figure*}
\begin{corollary}
	Let $u<q$ and let~\eqref{eq:GV-like-bound} hold, i.e., such that an $[n',k',d]_q$ code that contains the all-one vector as codeword exists.
	Then, there is a  ($u$,$\lfloor\tfrac{d-1}{2}\rfloor)$-PSMC of length $n'$.
\end{corollary}
\begin{proof}
	In \cite[Theorem~1]{haideralkim2019psmc}, it was shown that if the all-one vector is a codeword of a code with minimum distance $d$, then for $u <q$, there is a ($u$,$\lfloor\tfrac{d-1}{2}\rfloor)$-PSMC.
\end{proof}
Figure~\ref{fig3} compares the new GV-like bound to other well-known bounds. 
To explain how, let us first define the following:
\begin{itemize}
	\item Let $n_f$ and $k_f$ denote the $n$ and $k$ that is chosen to compute one point in the figure (here $n_f = 127, 121, 124,$ and $120$ and $k_f =$ values on the x-axis).
	\item Let $n$ and $k$ be the \emph{designed} length and dimension in the GV bound, i.e., what is used in the inequality \eqref{eq:GV-like-bound}.
	\item Let $n'$ and $k'$ be the actual length and dimension of the existing code.
\end{itemize} 
In the plots, we want to know for a fixed pair $[n_f,k_f]$, what the maximal minimum distance $d$ is that guarantees existence of a code. To have a fair comparison to another lower bound or construction, we chose the \emph{designed} $n$ and $k$ to our largest disadvantage, i.e.
$(n_f = n+1)$ which gives $(n = n_f-1)$ and 
$(k_f = k-d+2)$ which gives $(k = k_f+d-2)$. Next, we have to obtain $d$, so 
we plugged for $d=1,2,\dots$ his $n$ and $k$ (depending on $d$) into the bound in \eqref{eq:GV-like-bound} and computed the maximal $d$ that still satisfies the bound. Then we get a $d$ and Theorem~\ref{th11} tells us that there is an $[n',k',d]_q$ code with
$(n-d+2 \leq  n' \leq n+1 = n_f)$ and $(k_f = k-d+2 \leq k'\leq k+1)$.
In particular, the actual code has length $n'\leq n_f$ and dimension $k'\geq k_f$.
Hence, a code with exact parameters $[n_f, k_f, \geq d]_q$ exists. However, the values obtained for this $d$ are quite bad (the red curve).
For high dimensions $k_f$ and larger $q$, these values tend to be better as shown in (d), i.e. the red curve matches, and then is above the BCH curve.

\section{Conclusion}

We have proposed a new construction for combined masking of partially stuck-at-1 cells and error correction, by masking only binary classical stuck memory cells as proposed in \cite[Thorem~9]{wachterzeh2016codes}, with the error correction possibility similar to \cite[Theorem~4]{haideralkim2019psmc}.
Compared to \cite[Thorem~9]{wachterzeh2016codes}, the new code construction can correct errors in addition to masking. Furthermore, for specific examples on the code parameters, a higher amount of information symbols $k_1+l$ when $u \leq q \leq n$ can be stored compared to the code construction for masking and error correction in \cite[Theorem~4]{haideralkim2019psmc}.
Further, we have derived bounds on the required redundancy for a given number of partially stuck cells to mask and a given number of errors to correct. This includes a sphere-packing and a Gilbert--Varshamov-like bound.

Future work should calculate the capacity of a storage channel in which $u$ cells can be partially stuck at levels $s$ with probability $p$ and the rest $(u-n)$ healthy cells with probability $(1-p)$. If we assume disjoint case, $(u-n)$ cells have also crossover probability $\varepsilon$, i.e errors occur in $(u-n)$ cells only. 
Then we should compare it to the code rate of the new construction similar to \cite[Section \rom{9}]{wachterzeh2016codes}. 
\section*{Acknowledgement}
We would like to thank Ludo Tolhuizen for making us aware of a more accurate curve regarding GV-like bound by commenting on the previous version of this paper.
\bibliographystyle{IEEEtran}

\section {Appendix}
\emph{Example~1. Binary Codes for Masking and Correcting Partially Defect Memory}

\small{ Let $\ve{m} = \left(\begin{array}{rrrrrr}
	1 & 0 & 1 & \alpha &   1 +\alpha& 1
	\end{array}\right) \in \mathbb{F}_{2^2}^6$, $\ve{m}^\prime = \left(\begin{array}{rrrr}
	\alpha & 0 & \alpha & 0
	\end{array}\right) \in \mathcal{F}^{4}\subseteq \F_{2^2}$. Let $u=4$ so that $u_0=2 $. Since we mask $u$-PSMC by the mean of $u_0$-SMC, we need a code $\mathcal{C}_0$ of a minimum distance $d_0 = u_0 +1 = 3$ that its dual code is generated by $\ve{H}_0$. Let the $u$ stuck positions be $\phi_{0}=1, \phi_{1}=2, \phi_{2}=9$ and $\phi_{3}=14$.}
\small{Let $\ve{G}$ over $\mathbb{F}_4$ be a generator matrix of a code $\mathcal{C}$ with parameters $[15,11,3]_4$ from Theorem~\ref{thm9}:}
\\
\\
\scalebox{0.60}{ 
$ \ve{G} = \begin{bmatrix} &\textcolor{red}{\ve{H}_0}&& \textcolor{blue}0 \\&\ve{G}_1&&\textcolor{blue}{\vdots} \\&&&\textcolor{blue}0 \\ \textcolor{orange}1&\textcolor{orange}{\dots}&&\textcolor{orange}1   \end{bmatrix} =\\
\left(\begin{array}{rrrrrrrrrrrrrrr}
	\textcolor{red}1 & \textcolor{red}0 & \textcolor{red}0 & \textcolor{red}0 & \textcolor{red}1 & \textcolor{red}0 & \textcolor{red}1 & \textcolor{red}0 & \textcolor{red}1 & \textcolor{red}1 & \textcolor{red}1 & \textcolor{red}1 & \textcolor{red}0 & \textcolor{red}0 & \textcolor{blue}0 \\
	\textcolor{red}0 & \textcolor{red}1 & \textcolor{red}0 & \textcolor{red}0 & \textcolor{red}1 & \textcolor{red}1 & \textcolor{red}0 & \textcolor{red}1 & \textcolor{red}0 & \textcolor{red}1 & \textcolor{red}1 & \textcolor{red}1 & \textcolor{red}1 & \textcolor{red}0 & \textcolor{blue}0 \\
	\textcolor{red}0 & \textcolor{red}0 & \textcolor{red}1 & \textcolor{red}0 & \textcolor{red}0 & \textcolor{red}1 & \textcolor{red}1 & \textcolor{red}0 & \textcolor{red}1 & \textcolor{red}0 & \textcolor{red}1 & \textcolor{red}1 & \textcolor{red}1 & \textcolor{red}1 & \textcolor{blue}0 \\
	\textcolor{red}0 & \textcolor{red}0 & \textcolor{red}0 & \textcolor{red}1 & \textcolor{red}0 & \textcolor{red}0 & \textcolor{red}1 & \textcolor{red}1 & \textcolor{red}0 & \textcolor{red}1 & \textcolor{red}0 & \textcolor{red}1 & \textcolor{red}1 & \textcolor{red}1 & \textcolor{blue}0 \\
    0 & 0 & 0 & 0 & 1 & 0 & 0 & 0 & 0 & 0 & 0 & 1 & 0 & 1 & \textcolor{blue}0 \\
    0 & 0 & 0 & 0 & 0 & 1 & 0 & 0 & 0 & 0 & 1 & 1 & 1 & 0 & \textcolor{blue}0 \\
    0 & 0 & 0 & 0 & 0 & 0 & 1 & 0 & 0 & 0 & 0 & 1 & 1 & 1 & \textcolor{blue}0 \\
    0 & 0 & 0 & 0 & 0 & 0 & 0 & 1 & 0 & 0 & 1 & 1 & 1 & 1 & \textcolor{blue}0 \\
    0 & 0 & 0 & 0 & 0 & 0 & 0 & 0 & 1 & 0 & 1 & 0 & 1 & 1 & \textcolor{blue}0 \\
    0 & 0 & 0 & 0 & 0 & 0 & 0 & 0 & 0 & 1 & 1 & 0 & 0 & 1 & \textcolor{blue}0 \\
    
	\textcolor{orange}1 & \textcolor{orange}1 & \textcolor{orange}1 & \textcolor{orange}1 & \textcolor{orange}1 & \textcolor{orange}1 & \textcolor{orange}1 & \textcolor{orange}1 & \textcolor{orange}1 & \textcolor{orange}1 & \textcolor{orange}1 & \textcolor{orange}1 & \textcolor{orange}1 & \textcolor{orange}1 & \textcolor{orange}1
\end{array}\right)
$}
\\
\\
\small{ \emph{Encoding} follows Algorithm~\ref{a13} by plugging in the given values and matrices.\\
\small{ $\ve{G}^\prime$ over $\mathbb{F}_4$ of the code $\mathcal{C}$ of the parameters $[15,11,3]_4$ in \cite[Theroem~4]{haideralkim2019psmc} is:}

\scalebox{0.60}{
	$ \ve{G}^\prime = \begin{bmatrix} \ve{0}_{k_1 \times l} & \ve{I}_{k_1} &\ve{P}_{k_1\times r} \\ &\textcolor{red}{\ve{H}_0 }  &   \end{bmatrix} =
\left(\begin{array}{rrrrrrrrrrrrrrr}
	0 & 0 & 0 & 0 & 1 & 0 & 0 & 0 & 0 & 0 & 0 & 1 & 0 & 1 & 0 \\
	0 & 0 & 0 & 0 & 0 & 1 & 0 & 0 & 0 & 0 & 0 & 0 & 1 & 0 & 1 \\
	0 & 0 & 0 & 0 & 0 & 0 & 1 & 0 & 0 & 0 & 0 & 1 & 1 & 1 & 0 \\
	0 & 0 & 0 & 0 & 0 & 0 & 0 & 1 & 0 & 0 & 0 & 0 & 1 & 1 & 1 \\
	0 & 0 & 0 & 0 & 0 & 0 & 0 & 0 & 1 & 0 & 0 & 1 & 1 & 1 & 1 \\
	0 & 0 & 0 & 0 & 0 & 0 & 0 & 0 & 0 & 1 & 0 & 1 & 0 & 1 & 1 \\
	0 & 0 & 0 & 0 & 0 & 0 & 0 & 0 & 0 & 0 & 1 & 1 & 0 & 0 & 1 \\
	\textcolor{red}1 & \textcolor{red}0 & \textcolor{red}0 & \textcolor{red}0 & \textcolor{red}0 & \textcolor{red}0 & \textcolor{red}1 & \textcolor{red}1 & \textcolor{red}0 & \textcolor{red}1 & \textcolor{red}0 & \textcolor{red}1 & \textcolor{red}1 & \textcolor{red}1 & \textcolor{red}1 \\
	\textcolor{red}0 & \textcolor{red}1 & \textcolor{red}0 & \textcolor{red}0 & \textcolor{red}1 & \textcolor{red}1 & \textcolor{red}0 & \textcolor{red}1 & \textcolor{red}0 & \textcolor{red}1 & \textcolor{red}1 & \textcolor{red}1 & \textcolor{red}1 & \textcolor{red}0 & \textcolor{red}0 \\
	\textcolor{red}0 & \textcolor{red}0 & \textcolor{red}1 & \textcolor{red}0 & \textcolor{red}0 & \textcolor{red}1 & \textcolor{red}1 & \textcolor{red}0 & \textcolor{red}1 & \textcolor{red}0 & \textcolor{red}1 & \textcolor{red}1 & \textcolor{red}1 & \textcolor{red}1 & \textcolor{red}0 \\
	\textcolor{red}0 & \textcolor{red}0 & \textcolor{red}0 & \textcolor{red}1 & \textcolor{red}1 & \textcolor{red}0 & \textcolor{red}0 & \textcolor{red}1 & \textcolor{red}1 & \textcolor{red}0 & \textcolor{red}1 & \textcolor{red}0 & \textcolor{red}1 & \textcolor{red}1 & \textcolor{red}1
\end{array}\right)
$}\Large{.}
\\
\\
\small{It is good to mention that if we take the \emph{reduced echelon form} for both $\ve{G}$ and $\ve{G}^\prime$, the result is the same matrix $\ve{G}_e$. However, applying Theorem~\ref{thm9} gives higher rate compare to \cite[Theorem~4]{haideralkim2019psmc}.}

\end{document}